\documentclass[a4paper]{article}
\usepackage{geometry}
\usepackage{amsmath}
\usepackage{amssymb}
\usepackage{indentfirst}
\usepackage{mathdots}
\usepackage{cite}
\usepackage[colorlinks, linkcolor=blue, citecolor=blue]{hyperref}
\usepackage{xcolor}
\usepackage{mathrsfs}
\usepackage{booktabs}
\usepackage{enumitem}
\usepackage{authblk}
\usepackage{afterpage}
\usepackage{amsfonts,ntheorem}
\usepackage{threeparttable}
\usepackage{bm}
\usepackage{subfigure, epsfig}
\usepackage{setspace}
\usepackage{makecell}
\numberwithin{equation}{section}

\usepackage{algorithm}
\usepackage{algpseudocode}

\makeatletter

\newcommand{\Rmnum}[1]{\expandafter\@slowromancap\romannumeral #1@}
\makeatother

\newtheorem{theorem}{Theorem}[section]
\newtheorem{lemma}[theorem]{Lemma}
\newtheorem{corollary}[theorem]{Corollary}

{ 
\theoremstyle{plain}
\theorembodyfont{\normalfont} 

\newtheorem{definition}[theorem]{Definition}

\newtheorem{remark}[theorem]{Remark}

}

\newenvironment{proof}{\noindent{\textbf{\emph{Proof.}}}}

\allowdisplaybreaks[4]

\definecolor{mydarkgreen}{RGB}{0,100,0}

\begin{document}

\title{\Large The $\sigma$ hulls of matrix-product codes and related entanglement-assisted quantum error-correcting codes}
\author[1]{{\small{Meng Cao}} \thanks{\footnotesize{Beijing Institute of Mathematical Sciences and Applications, 101408, Beijing, China
(E-mail: mengcaomath@126.com)}
 }}
\renewcommand*{\Affilfont}{\small\it}

\date{}
\maketitle

{\linespread{1.4}{

\vskip -7mm

\noindent {\small{{\bfseries{Abstract:}}
Let $\mathrm{SLAut}(\mathbb{F}_{q}^{n})$ denote the group of all semilinear isometries on $\mathbb{F}_{q}^{n}$, where $q=p^{e}$ is a prime power.
Matrix-product (MP) codes are a class of long classical codes generated by combining several commensurate classical codes with a defining matrix.
We give an explicit formula for calculating the dimension of the $\sigma$ hull of a MP code.
As a result, we give necessary and sufficient conditions for the MP codes to be $\sigma$ dual-containing and $\sigma$ self-orthogonal. We prove that $\mathrm{dim}_{\mathbb{F}_{q}}(\mathrm{Hull}_{\sigma}(\mathcal{C}))=\mathrm{dim}_{\mathbb{F}_{q}}(\mathrm{Hull}_{\sigma}(\mathcal{C}^{\bot_{\sigma}}))$.
We prove that for any integer $h$ with $\mathrm{max}\{0,k_{1}-k_{2}\}\leq h\leq \mathrm{dim}_{\mathbb{F}_{q}}(\mathcal{C}_{1}\cap\mathcal{C}_{2}^{\bot_{\sigma}})$, there exists a linear code $\mathcal{C}_{2,h}$ monomially equivalent to $\mathcal{C}_{2}$ such that $\mathrm{dim}_{\mathbb{F}_{q}}(\mathcal{C}_{1}\cap\mathcal{C}_{2,h}^{\bot_{\sigma}})=h$, where $\mathcal{C}_{i}$ is an $[n,k_{i}]_{q}$ linear code for $i=1,2$.
We show that given an $[n,k,d]_{q}$ linear code $\mathcal{C}$, there exists a monomially equivalent $[n,k,d]_{q}$ linear code $\mathcal{C}_{h}$, whose $\sigma$ dual code has minimum distance $d'$, such that there exist an $[[n,k-h,d;n-k-h]]_{q}$ EAQECC and an $[[n,n-k-h,d';k-h]]_{q}$ EAQECC for every integer $h$ with $0\leq h\leq \mathrm{dim}_{\mathbb{F}_{q}}(\mathrm{Hull}_{\sigma}(\mathcal{C}))$. Based on this result, we present a general construction method for deriving EAQECCs with flexible parameters from MP codes related to $\sigma$ hulls.
}}

\vspace{10pt}


\vspace{6pt}
}

\section{Introduction}\label{section1}

Matrix-product (MP) codes (\cite{Blackmore2001Matrix-product,O2002Note}) are a class of long classical codes generated by combining several commensurate
classical codes with a defining matrix. Here, we adopt the methodology proposed by Blackmore and Norton \cite{Blackmore2001Matrix-product} to introduce
MP codes over finite fields.

Let $q=p^{e}$ be a prime power and let $\ell$ be an integer with $0\leq\ell\leq e-1$.
Let $\mathbb{F}_{q}$ be the finite field with $q$ elements and let $\mathbb{F}_{q}^{\ast}:=\mathbb{F}_{q}\backslash\{0\}$, where $q$ is a prime power.
Let $\mathbb{F}_{q}^{k\times t}$ denote the set of all $k\times t$ matrices over $\mathbb{F}_{q}$.
We denote a linear code over $\mathbb{F}_{q}$ with length $n$, dimension $k$, and minimum distance $d$ as $[n,k,d]_q$.
Let $\mathcal{C}_{1},\mathcal{C}_{2},\ldots,\mathcal{C}_{k}$ be linear codes of length $n$ over $\mathbb{F}_{q}$.
Consider a matrix $A=(a_{i,j})\in \mathbb{F}_{q}^{k\times t}$ with $k\leq t$.
We define a \textbf{matrix-product (MP) code} as
$\mathcal{C}(A):=[\mathcal{C}_{1},\mathcal{C}_{2},\ldots,\mathcal{C}_{k}]\cdot A$, which is the set of all matrix-products
$[\mathbf{c}_{1},\mathbf{c}_{2},\ldots,\mathbf{c}_{k}]\cdot A$.
Here, $\mathbf{c}_{i}=(c_{1,i},c_{2,i},\ldots,c_{n,i})^{\top}\in\mathcal{C}_{i}$ for $i=1,2,\ldots,k$.
In this context, $A$ is referred to as the \textbf{defining matrix} of $\mathcal{C}(A)$, and $\mathcal{C}_{1},\mathcal{C}_{2},\ldots,\mathcal{C}_{k}$ are called the \textbf{constituent codes} of $\mathcal{C}(A)$.
A standard codeword $\mathbf{c}=[\mathbf{c}_{1},\mathbf{c}_{2},\ldots,\mathbf{c}_{k}]\cdot A$ of $\mathcal{C}(A)$ is represented by an $n\times t$ matrix as follows:
\begin{align*}
\mathbf{c}=\left(
\begin{array}{cccc}
\sum_{i=1}^k c_{1,i}a_{i,1}& \sum_{i=1}^k c_{1,i}a_{i,2}& \cdots &\sum_{i=1}^k c_{1,i}a_{i,t}\\
\sum_{i=1}^k c_{2,i}a_{i,1}& \sum_{i=1}^k c_{2,i}a_{i,2}& \cdots &\sum_{i=1}^k c_{2,i}a_{i,t}\\
\vdots&\vdots&\ddots&\vdots\\
\sum_{i=1}^k c_{n,i}a_{i,1}& \sum_{i=1}^k c_{n,i}a_{i,2}& \cdots &\sum_{i=1}^k c_{n,i}a_{i,t}\\
\end{array}\right).
\end{align*}
By reading the entries of the $n \times t$ matrix  above in column-major order, any codeword of $\mathcal{C}(A)$ can be represented as a row vector of length $tn$, meaning that
$\mathbf{c}=\big[\sum_{i=1}^k a_{i,1}\mathbf{c}_{i},\sum_{i=1}^k a_{i,2}\mathbf{c}_{i},\ldots,\sum_{i=1}^k a_{i,t}\mathbf{c}_{i}\big]$,
where $\mathbf{c}_{i}=(c_{1,i},c_{2,i},\ldots,c_{n,i})\in\mathcal{C}_{i}$ is a $1\times n$ row vector for $i=1,2,\ldots,k$.
It was shown in \cite[p.480]{Blackmore2001Matrix-product} that
\begin{align}\label{equation1.1}
G:=\left(
\begin{array}{cccc}
a_{1,1}G_{1}& a_{1,2}G_{1}&  \cdots &a_{1,t}G_{1}\\
a_{2,1}G_{2}& a_{2,2}G_{2}& \cdots &a_{2,t}G_{2}\\
\vdots&\vdots&\ddots&\vdots\\
a_{k,1}G_{k}& a_{k,2}G_{k}& \cdots &a_{k,t}G_{k}  \\
\end{array}\right)
\end{align}
is a generator matrix of $C(A)$ if $G_{i}$ is a generator matrix of $\mathcal{C}_{i}$ for $i=1,2,\ldots,k$.

Denote by $\mathrm{SLAut}(\mathbb{F}_{q}^{n})$ the group consisting of all semilinear isometries on $\mathbb{F}_{q}^{n}$.
Each element $\sigma\in\mathrm{SLAut}(\mathbb{F}_{q}^{n})$ can be expressed as $\sigma=(\tau,\pi_{s})$,
which satisfies $\sigma(\mathbf{c})=\pi_{s}(\mathbf{c})M_{\tau}$ for any $\mathbf{c}\in\mathbb{F}_{q}^{n}$.
Here, the permutation $\tau$ in the symmetric group of order $n$ corresponds to an $n\times n$ monomial matrix $M_{\tau}$,
and $\pi_{s}$ denotes an automorphism of $\mathbb{F}_{q}$ with $1\leq s\leq e$ satisfying $\pi_{s}(a)=a^{p^{s}}$ for any $a\in\mathbb{F}_{q}$.
Recently, Carlet \emph{et al}. \cite{Carlet2019On} introduced the concept of $\sigma$ inner product for $\sigma\in\mathrm{SLAut}(\mathbb{F}_{q}^{n})$,
generalizing the Euclidean inner product, Hermitian inner product and $\ell$-Galois inner product.
From the concept of $\sigma$ inner product, it is clear that the $\sigma$ dual $\mathcal{C}^{\bot_{\sigma}}$ of a linear code $\mathcal{C}$
extends the Euclidean dual $\mathcal{C}^{\bot}$, Hermitian dual $\mathcal{C}^{\bot_{H}}$ and $\ell$-Galois dual $\mathcal{C}^{\bot_{\ell}}$ of it.

In this paper, we give an explicit formula for calculating the dimension of the $\sigma$ hull of a MP code.
As a result, we give necessary and sufficient conditions for the MP codes to be $\sigma$ dual-containing and $\sigma$ self-orthogonal. We prove that $\mathrm{dim}_{\mathbb{F}_{q}}(\mathrm{Hull}_{\sigma}(\mathcal{C}))=\mathrm{dim}_{\mathbb{F}_{q}}(\mathrm{Hull}_{\sigma}(\mathcal{C}^{\bot_{\sigma}}))$.
We prove that for any integer $h$ with $\mathrm{max}\{0,k_{1}-k_{2}\}\leq h\leq \mathrm{dim}_{\mathbb{F}_{q}}(\mathcal{C}_{1}\cap\mathcal{C}_{2}^{\bot_{\sigma}})$, there exists a linear code $\mathcal{C}_{2,h}$ monomially equivalent to $\mathcal{C}_{2}$ such that $\mathrm{dim}_{\mathbb{F}_{q}}(\mathcal{C}_{1}\cap\mathcal{C}_{2,h}^{\bot_{\sigma}})=h$, where $\mathcal{C}_{i}$ is an $[n,k_{i}]_{q}$ linear code for $i=1,2$.
We show that given an $[n,k,d]_{q}$ linear code $\mathcal{C}$, there exists a monomially equivalent $[n,k,d]_{q}$ linear code $\mathcal{C}_{h}$, whose $\sigma$ dual code has minimum distance $d'$, such that there exist an $[[n,k-h,d;n-k-h]]_{q}$ EAQECC and an $[[n,n-k-h,d';k-h]]_{q}$ EAQECC for every integer $h$ with $0\leq h\leq \mathrm{dim}_{\mathbb{F}_{q}}(\mathrm{Hull}_{\sigma}(\mathcal{C}))$. Based on this result, we present a general construction method for deriving EAQECCs with flexible parameters from MP codes related to $\sigma$ hulls.
To illustrate the advantages of the EAQECCs constructed using our method, we provide some new EAQECCs that are not covered in recent literature.

\section{Preliminaries}\label{section2}

Denote by $[n,k,d]_{q}$ a classical linear code over $\mathbb{F}_{q}$ with length $n$, dimension $k$ and minimum distance $d$.
The minimum distance $d$ of a linear code satisfies the well-known Singleton bound $d\leq n-k+1$.
If $d=n-k+1$, then such a linear code is called a \textbf{maximum distance separable (MDS) code}.

Following the notations introduced in \cite{Carlet2019On},
we call a mapping $\sigma: \mathbb{F}_{q}^{n}\rightarrow\mathbb{F}_{q}^{n}$ an \textbf{isometry} if $d_{H}\big(\sigma(\mathbf{u}),\sigma(\mathbf{v})\big)$
$=d_{H}(\mathbf{u},\mathbf{v})$
for any $\mathbf{u},\mathbf{v}\in\mathbb{F}_{q}^{n}$,
where $d_{H}(\mathbf{u},\mathbf{v})$ is the Hamming distance of $\mathbf{u}$ and $\mathbf{v}$.
All isometries on $\mathbb{F}_{q}^{n}$ form a group, which is denoted by $\mathrm{Aut}(\mathbb{F}_{q}^{n})$.
Two codes $\mathcal{C}$ and $\mathcal{C}'$ are called \textbf{isometric} if $\sigma(\mathcal{C})=\mathcal{C}'$ for some $\sigma\in\mathrm{Aut}(\mathbb{F}_{q}^{n})$.
If an isometry is linear, then it is called a linear isometry.
Let $\mathrm{MAut}(\mathbb{F}_{q}^{n})$ denote the \textbf{monomial group} consisting of
the set of transforms with $n\times n$ monomial matrices over $\mathbb{F}_{q}$.
Notice that the group of all linear isometries of $\mathbb{F}_{q}^{n}$ corresponds to $\mathrm{MAut}(\mathbb{F}_{q}^{n})$.

For isometries that map subspaces onto subspaces, it is shown in \cite{Betten2006Error,Sendrier2013The}
that when $n\geq 3$, such isometries are exactly the semilinear mappings of the form
{\setlength\abovedisplayskip{0.15cm}
\setlength\belowdisplayskip{0.15cm}
\begin{align*}
\begin{split}
\sigma=(\tau,\pi):\ &\mathbb{F}_{q}^{n}\longrightarrow\mathbb{F}_{q}^{n}\\
&\ \ \mathbf{c}\longmapsto\tau\big(\pi(\mathbf{c})\big)
\end{split}
\end{align*}}with $\pi(\mathbf{c}):=\big(\pi(c_{1}),\pi(c_{2}),\ldots,\pi(c_{n})\big)$ for $\mathbf{c}=(c_{1},c_{2},\ldots,c_{n})\in\mathbb{F}_{q}^{n}$,
where $\tau$ is a linear isometry and $\pi$ is an automorphism of $\mathbb{F}_{q}$.
We denote by $\mathrm{SLAut}(\mathbb{F}_{q}^{n})$ the group of all semilinear isometries on $\mathbb{F}_{q}^{n}$.

\begin{remark}\label{remark2.1}
For any $\sigma=(\tau,\pi)\in \mathrm{SLAut}(\mathbb{F}_{q}^{n})$, the above definition implies that there exists an $n\times n$ monomial matrix $M_{\tau}$
over $\mathbb{F}_{q}$ which corresponds to $\tau$ such that
{\setlength\abovedisplayskip{0.15cm}
\setlength\belowdisplayskip{0.15cm}
\begin{equation}\label{eq2.1}
\sigma(\mathbf{c})=\tau\big(\pi(\mathbf{c})\big)=\pi(\mathbf{c})M_{\tau}
\end{equation}}for any $\mathbf{c}\in\mathbb{F}_{q}^{n}$.
Here, we can write the monomial matrix $M_{\tau}$ as $M_{\tau}=D_{\tau}P_{\tau}$, where $D_{\tau}$ is an invertible diagonal matrix,
and $P_{\tau}$ is a permutation matrix under $\tau$ such that the $\tau(i)$-th row of $P_{\tau}$ is exactly the $i$-th row of the identity matrix $I_{n}$
for $i=1,2,\ldots,n$. Consequently, it is clear that
{\setlength\abovedisplayskip{0.15cm}
\setlength\belowdisplayskip{0.15cm}
\begin{equation*}
(t_{1},t_{2},\ldots,t_{n})P_{\tau}=\big(t_{\tau(1)},t_{\tau(2)},\ldots,t_{\tau(n)}\big),
\end{equation*}
\begin{equation}\label{eq2.2}
(t_{1},t_{2},\ldots,t_{n})P_{\tau}^{T}=\big(t_{\tau^{-1}(1)},t_{\tau^{-1}(2)},\ldots,t_{\tau^{-1}(n)}\big).
\end{equation}}
\end{remark}

Carlet \emph{et al}. \cite{Carlet2019On} introduced the following concepts on the $\sigma$ inner product and $\sigma$ dual of linear codes.

\begin{definition}\label{definition2.2}
(\!\!\cite{Carlet2019On})
Let $q=p^{e}$ be a prime power and let $\mathcal{C}$ be a linear code of length $n$ over $\mathbb{F}_{q}$. For $\sigma=(\tau,\pi)\in\mathrm{SLAut}(\mathbb{F}_{q}^{n})$, the \emph{$\sigma$ inner product} of
$\mathbf{a}\in\mathbb{F}_{q}^{n}$ and $\mathbf{b}\in\mathbb{F}_{q}^{n}$ is defined as
{\setlength\abovedisplayskip{0.15cm}
\setlength\belowdisplayskip{0.15cm}
\begin{equation*}
\langle\mathbf{a},\mathbf{b}\rangle_{\sigma}=\sum_{i=1}^{n}a_{i}c_{i},
\end{equation*}}where $\mathbf{a}=(a_{1},a_{2},\ldots,a_{n})$ and $\sigma(\mathbf{b})=(c_{1},c_{2},\ldots,c_{n})$.
The \emph{$\sigma$ dual} of $\mathcal{C}$ is defined as
{\setlength\abovedisplayskip{0.15cm}
\setlength\belowdisplayskip{0.15cm}
\begin{equation*}
\mathcal{C}^{\bot_{\sigma}}=\{\mathbf{a}\in\mathbb{F}_{q}^{n}|\ \langle\mathbf{a},\mathbf{b}\rangle_{\sigma}=0 \ \mathrm{for} \ \mathrm{all} \  \mathbf{b}\in\mathcal{C}\}.
\end{equation*}}
\end{definition}

Denote $\sigma(\mathcal{C}):=\{\sigma(\mathbf{c})|\ \mathbf{c}\in\mathcal{C}\}$.
Then $\mathcal{C}^{\bot_{\sigma}}=\big(\sigma(\mathcal{C})\big)^{\bot_{E}}$ (see \cite{Carlet2019On}).

\begin{remark}\label{remark2.3}
If $M_{\tau}=I_{n}$ and $\pi=\pi_{e-\ell}$ in Eq. (\ref{eq2.1}), where $\pi_{e-\ell}(a):=a^{p^{e-\ell}}$ with $0\leq \ell\leq e-1$
for each $a\in\mathbb{F}_{q}$,
then $\langle\mathbf{a},\mathbf{b}\rangle_{\sigma}=\langle\mathbf{a},\mathbf{b}\rangle_{e-\ell}:=\sum_{i=1}^{n}a_{i}b_{i}^{p^{e-\ell}}$,
which is the $(e-\ell)$-Galois inner product \cite{Fan2017Galois} of $\mathbf{a}$ and $\mathbf{b}$; $\mathcal{C}^{\bot_{\sigma}}=\mathcal{C}^{\bot_{\ell}}$,
which is the $\ell$-Galois dual \cite{Fan2017Galois} of $\mathcal{C}$. Further,
if $\ell=0$, then $\langle\mathbf{a},\mathbf{b}\rangle_{\sigma}=\langle\mathbf{a},\mathbf{b}\rangle_{E}:=\sum_{i=1}^{n}a_{i}b_{i}$,
which is the Euclidean inner product of $\mathbf{a}$ and $\mathbf{b}$; $\mathcal{C}^{\bot_{\sigma}}=\mathcal{C}^{\bot_{E}}$, which is the Euclidean dual of $\mathcal{C}$.
If $\ell=\frac{e}{2}$ for even $e$, then $\langle\mathbf{a},\mathbf{b}\rangle_{\sigma}=\langle\mathbf{a},\mathbf{b}\rangle_{H}:=\sum_{i=1}^{n}a_{i}b_{i}^{\sqrt{q}}$,
which is the Hermitian inner product of $\mathbf{a}$ and $\mathbf{b}$; $\mathcal{C}^{\bot_{\sigma}}=\mathcal{C}^{\bot_{H}}$, which is the Hermitian dual of $\mathcal{C}$.
\end{remark}

We fix the following notations in the rest of this paper.
\begin{itemize}
\setlength{\itemsep}{1pt}
\setlength{\parsep}{1pt}
\setlength{\parskip}{1pt}
\item $q=p^{e}$ is a prime power, where $p$ is a prime number and $e$ is a positive integer.

\item $\mathbb{F}_{q}^{\ast}:=\mathbb{F}_{q}\backslash \{0\}$.

\item Let $\mathbb{F}_{q}^{s\times t}$ denote the set of all $s\times t$ matrices over $\mathbb{F}_{q}$.

\item $A^{\ast}:=\big[a_{i,j}^{\sqrt{q}}\big]$ and $A^{\dag}:=\big[a_{j,i}^{\sqrt{q}}\big]$ for any matrix $A=[a_{i,j}]$ over $\mathbb{F}_{q}$.

\item $\mathcal{C}A:=\{\mathbf{c}A|\mathbf{c}\in\mathcal{C}\}$, where $\mathcal{C}$ is a linear code of length $n$ over $\mathbb{F}_{q}$ and
$A\in\mathbb{F}_{q}^{n\times n}$.

\item $\pi_{s}$ denotes the automorphism of $\mathbb{F}_{q}$  satisfying $\pi_{s}(a)=a^{p^{s}}$ for any $a\in\mathbb{F}_{q}$,
where $1\leq s\leq e$.

\item $\pi_{s}(A):=[\pi_{s}(a_{i,j})]$ for any matrix $A=[a_{i,j}]$ over $\mathbb{F}_{q}$.

\item  For $\sigma=(\tau,\pi_{s})\in\mathrm{SLAut}(\mathbb{F}_{q}^{n})$,
$\tau$ corresponds to a monomial matrix $M_{\tau}=D_{\tau}P_{\tau}$ with $D_{\tau}$ being an invertible diagonal matrix
and $P_{\tau}$ being a permutation matrix under $\tau$, that is to say, $\tau(\mathbf{u})=\mathbf{u}M_{\tau}$ for any $\mathbf{u}\in\mathbb{F}_{q}^{n}$.
Therefore, $\sigma(\mathbf{c})=\tau\big(\pi_{s}(\mathbf{c})\big)=\pi_{s}(\mathbf{c})M_{\tau}$ for any $\mathbf{c}\in\mathbb{F}_{q}^{n}$
(see also Eq. (\ref{eq2.1})).

\item $\mathrm{Hull}_{\sigma}(\mathcal{C}):=\mathcal{C}\cap\mathcal{C}^{\bot_{\sigma}}$ is called the \emph{$\sigma$ hull} of the linear code $\mathcal{C}$,
where $\sigma\in \mathrm{SLAut}(\mathbb{F}_{q}^{n})$.

\item $\mathrm{Hull}_{\ell}(\mathcal{C}):=\mathcal{C}\cap\mathcal{C}^{\bot_{\ell}}$ is called the \emph{$\ell$-Galois hull} of the linear code $\mathcal{C}$,
where $0\leq \ell\leq e-1$.

\item $\mathrm{Hull}(\mathcal{C}):=\mathcal{C}\cap\mathcal{C}^{\bot}$ is called the \emph{Euclidean hull} of the linear code $\mathcal{C}$.

\item $\mathrm{Hull}_{H}(\mathcal{C}):=\mathcal{C}\cap\mathcal{C}^{\bot_{H}}$ is called the \emph{Hermitian hull} of the linear code $\mathcal{C}$.

\end{itemize}

\begin{lemma}\label{lemma1}
(\!\! \cite[Lemma 3.1]{Cao2024Onthe})
Let $\mathcal{C}$ be an $[n,k]_{q}$ linear code with parity check matrix $H$. Set $\sigma=(\tau,\pi_{s})\in\mathrm{SLAut}(\mathbb{F}_{q}^{n})$,
where $\tau$ corresponds to a monomial matrix $M_{\tau}=DP_{\tau}\in\mathbb{F}_{q}^{n\times n}$ and $1\leq s\leq e$. Then
$\pi_{s}(H)(M_{\tau}^{-1})^{\top}$ is a generator matrix of $\mathcal{C}^{\bot_{\sigma}}$.
\end{lemma}

\begin{lemma}\label{lemma2}
(\!\! \cite[Theorem 2.1]{Guenda2020Linear})
Let $\mathcal{C}_{i}$ be an $[n,k_{i}]_{q}$ linear code with generator matrix $G_{i}$ and parity check matrix $H_{i}$ for $i=1,2$. Then
\vspace{-5pt}
\begin{itemize}
\item [(1)] $\mathrm{dim}_{\mathbb{F}_{q}}(\mathcal{C}_{1}\cap\mathcal{C}_{2})=k_{2}-\mathrm{rank}(H_{1}G_{2}^{\top})$;

\vspace{-4pt}

\item [(2)] $\mathrm{dim}_{\mathbb{F}_{q}}(\mathcal{C}_{1}\cap\mathcal{C}_{2})=k_{1}-\mathrm{rank}(G_{1}H_{2}^{\top})$.
\end{itemize}
\end{lemma}

\begin{lemma}\label{lemma2.6}
(\!\! \cite[Theorem 3.3]{Anderson2024Relative})
Let $q>2$ and let $\mathcal{C}_{i}$ be an $[n,k_{i}]_{q}$ linear code for $i=1,2$. Then for any integer $h$ with
$\mathrm{max}\{0,k_{1}-k_{2}\}\leq h\leq \mathrm{dim}_{\mathbb{F}_{q}}(\mathcal{C}_{1}\cap\mathcal{C}_{2}^{\bot})$, there exists a linear code
$\mathcal{C}_{2,h}$ equivalent to $\mathcal{C}_{2}$ such that
$\mathrm{dim}_{\mathbb{F}_{q}}(\mathcal{C}_{1}\cap\mathcal{C}_{2,h}^{\bot})=h$.
\end{lemma}

\section{Some properties of MP codes with $\sigma$ inner product}

We present a formula for calculating the dimension of the $\sigma$ hull of an MP code in the following theorem.

\begin{theorem}\label{theorem3.1}
Let $\mathcal{C}(A)=[\mathcal{C}_{1},\mathcal{C}_{2},\ldots,\mathcal{C}_{k}]\cdot A$, where $\mathcal{C}_{i}$ is an $[n,t_{i}]_{q}$ code for $i=1,2,\ldots,k$ and $A\in\mathbb{F}_{q}^{k\times k}$.
Let $\sigma=(\tau,\pi_{s})\in\mathrm{SLAut}(\mathbb{F}_{q}^{kn})$ and $\widetilde{\sigma}=(\widetilde{\tau},\pi_{s})\in\mathrm{SLAut}(\mathbb{F}_{q}^{n})$,
where $1\leq s\leq e$, $\widetilde{\tau}$ corresponds to a $\widetilde{\tau}$-monomial matrix $M_{\widetilde{\tau}}\in\mathbb{F}_{q}^{n}$,
and $\tau$ corresponds to a $\tau$-monomial matrix $M_{\tau}=M_{\widehat{\tau}}\otimes M_{\widetilde{\tau}}\in\mathbb{F}_{q}^{kn}$ for $\widehat{\tau}$-monomial matrix $M_{\widehat{\tau}}\in\mathbb{F}_{q}^{k}$.
If $\pi_{s}(A)M_{\widehat{\tau}}A^{\top}$ is $\varrho$-monomial for some permutation $\varrho\in S_{k}$, then
\vspace{-4pt}
\begin{align*}
\mathrm{dim}_{\mathbb{F}_{q}}(\mathrm{Hull}_{\sigma}(\mathcal{C}(A)))
=\sum_{i=1}^{k}\mathrm{dim}_{\mathbb{F}_{q}}(\mathcal{C}_{i}\cap\mathcal{C}_{\rho(i)}^{\bot_{\widetilde{\sigma}}})
\end{align*}
\end{theorem}

\begin{proof}
Denote by $\mathbf{a}_{i}$ the $i$-th row of $A=(a_{i,j})\in\mathbb{F}_{q}^{k\times k}$ for $i=1,2,\ldots,k$. Let $G_{i}$ be a generator matrix of $\mathcal{C}_{i}$ for $i=1,2,\ldots,k$. By Eq. (\ref{equation1.1}), we know that $\mathcal{C}(A)$ has a generator matrix with the form
\vspace{-4pt}
\begin{align}\label{equation3.1}
G=\left(
\begin{array}{c}
\mathbf{a}_{1}\otimes G_{1}\\
\vdots\\
\mathbf{a}_{k}\otimes G_{k}\\
\end{array}\right).
\end{align}

As $\pi_{s}(A)M_{\widehat{\tau}}A^{\top}$ is $\varrho$-monomial for some permutation $\varrho\in S_{k}$, we can represent it as
$\pi_{s}(A)M_{\widehat{\tau}}A^{\top}=D_{\varrho}P_{\varrho}$, where $D_{\varrho}=\mathrm{diag}(\alpha_{1},\alpha_{2},\ldots,\alpha_{k})$ is invertible diagonal
with $\alpha_{i}\in\mathbb{F}_{q}^{\ast}$ for $i=1,2,\ldots,k$, and $P_{\varrho}$ is a permutation matrix under $\varrho$.
Then for any $1\leq i, j\leq k$, we have
\vspace{-4pt}
\begin{equation*}
\pi_{s}(\mathbf{a}_{i})M_{\widehat{\tau}}\mathbf{a}_{j}^{\top}
=\left\{
\begin{array}{cc}
\alpha_{j}, & \text{if } i=\varrho(j); \\
0, & \text{if } i\neq\varrho(j).
\end{array}
\right.
\end{equation*}
Hence, for any $1\leq i, j\leq k$, it follows that
\vspace{-4pt}
\begin{align*}
\pi_{s}(\mathbf{a}_{i}\otimes G_{i})M_{\tau}(\mathbf{a}_{j}^{\top}\otimes G_{j}^{\top})
&=(\pi_{s}(\mathbf{a}_{i})\otimes\pi_{s}(G_{i}))(M_{\widehat{\tau}}\otimes M_{\widetilde{\tau}})(\mathbf{a}_{j}^{\top}\otimes G_{j}^{\top})\\
&=(\pi_{s}(\mathbf{a}_{i})M_{\widehat{\tau}}\mathbf{a}_{j}^{\top})(\pi_{s}(G_{i})M_{\widetilde{\tau}}G_{j}^{\top})\\
&=\left\{
\begin{array}{cc}
\alpha_{j}\pi_{s}(G_{\varrho(j)})M_{\widetilde{\tau}}G_{j}^{\top}, & \text{if } i=\varrho(j); \\
0, & \text{if } i\neq\varrho(j).
\end{array}
\right.
\end{align*}

Combining this with Eq. (\ref{equation3.1}), we derive that
\vspace{-4pt}
\begin{align*}
\pi_{s}(G)M_{\tau}G^{\top}
&=\left(
\begin{array}{c}
\pi_{s}(\mathbf{a}_{1}\otimes G_{1})\\
\vdots\\
\pi_{s}(\mathbf{a}_{k}\otimes G_{k})\\
\end{array}\right)M_{\tau}(\mathbf{a}_{1}^{\top}\otimes G_{1}^{\top},\ldots,\mathbf{a}_{k}^{\top}\otimes G_{k}^{\top})\\
&=\left(
\begin{array}{ccc}
\pi_{s}(\mathbf{a}_{1}\otimes G_{1})M_{\tau}(\mathbf{a}_{1}^{\top}\otimes G_{1}^{\top})&\cdots&\pi_{s}(\mathbf{a}_{1}\otimes G_{1})M_{\tau}(\mathbf{a}_{k}^{\top}\otimes G_{k}^{\top})\\
\vdots&\ddots&\vdots\\
\pi_{s}(\mathbf{a}_{k}\otimes G_{k})M_{\tau}(\mathbf{a}_{1}^{\top}\otimes G_{1}^{\top})&\cdots&\pi_{s}(\mathbf{a}_{k}\otimes G_{k})M_{\tau}(\mathbf{a}_{k}^{\top}\otimes G_{k}^{\top})\\
\end{array}\right)\\
&=\left(
\begin{array}{ccc}
\vdots&\cdots&\vdots\\
\alpha_{1}\pi_{s}(G_{\varrho(1)})M_{\widetilde{\tau}}G_{1}^{\top}&\cdots&\vdots\\
\vdots&\ddots&\alpha_{k}\pi_{s}(G_{\varrho(k)})M_{\widetilde{\tau}}G_{k}^{\top}\\
\vdots&\cdots&\vdots\\
\end{array}\right),
\end{align*}
which, together with Lemma \ref{lemma3.1}, illustrates that
\vspace{-4pt}
\begin{align*}
\mathrm{rank}(\pi_{s}(G)M_{\tau}G^{\top})
=\sum_{i=1}^{k}\mathrm{rank}(\pi_{s}(G_{\varrho(i)})M_{\widetilde{\tau}}G_{i}^{\top})
=\sum_{i=1}^{k}(t_{i}-\mathrm{dim}_{\mathbb{F}_{q}}(\mathcal{C}_{i}\cap\mathcal{C}_{\rho(i)}^{\bot_{\widetilde{\sigma}}})).
\end{align*}

Finally, by Corollary \ref{corollary3.2}, we deduce that
\vspace{-4pt}
\begin{align*}
\mathrm{dim}_{\mathbb{F}_{q}}(\mathrm{Hull}_{\sigma}(\mathcal{C}(A)))
=\sum_{i=1}^{k}t_{i}-\mathrm{rank}(\pi_{s}(G)M_{\tau}G^{\top})
=\sum_{i=1}^{k}\mathrm{dim}_{\mathbb{F}_{q}}(\mathcal{C}_{i}\cap\mathcal{C}_{\rho(i)}^{\bot_{\widetilde{\sigma}}}),
\end{align*}
which completes the proof. $\hfill\square$
\end{proof}

\vspace{6pt}

In the following theorem, we give necessary and sufficient condition for the MP codes to be $\sigma$ dual-containing codes and $\sigma$ self-orthogonal codes.

\begin{theorem}\label{theorem3.2}
With the notations defined as in Theorem \ref{theorem3.1}. Then the following statements hold.
\vspace{-5pt}
\begin{itemize}
\item [(1)] $\mathcal{C}(A)$ is $\sigma$ dual-containing if and only if $\mathcal{C}_{\rho(i)}^{\bot_{\widetilde{\sigma}}}\subseteq \mathcal{C}_{i}$ for $i=1,2,\ldots,k$.

\vspace{-4pt}

\item [(2)] $\mathcal{C}(A)$ is $\sigma$ self-orthogonal if and only if $\mathcal{C}_{i}\subseteq\mathcal{C}_{\rho(i)}^{\bot_{\widetilde{\sigma}}}$ for $i=1,2,\ldots,k$.
\end{itemize}
\end{theorem}

\begin{proof}
According to Remark 4.5 2) of \cite{Cao2024Onthe}, we have
\vspace{-4pt}
\begin{align*}
\mathrm{dim}_{\mathbb{F}_{q}}(\mathcal{C}(A)^{\bot_{\sigma}})
=\sum_{i=1}^{k}\mathrm{dim}_{\mathbb{F}_{q}}(\mathcal{C}_{i}^{\bot_{\widetilde{\sigma}}})
=\sum_{i=1}^{k}\mathrm{dim}_{\mathbb{F}_{q}}(\mathcal{C}_{\varrho(i)}^{\bot_{\widetilde{\sigma}}}).
\end{align*}
Combining it with Theorem \ref{theorem3.1}, we obtain
\vspace{-4pt}
\begin{align*}
\mathcal{C}(A)\ \text{is $\sigma$ dual-containing}
&\Longleftrightarrow\mathrm{dim}_{\mathbb{F}_{q}}(\mathrm{Hull}_{\sigma}(\mathcal{C}(A)))=\mathrm{dim}_{\mathbb{F}_{q}}(\mathcal{C}(A)^{\bot_{\sigma}}).\\
&\Longleftrightarrow\sum_{i=1}^{k}\mathrm{dim}_{\mathbb{F}_{q}}(\mathcal{C}_{i}\cap\mathcal{C}_{\rho(i)}^{\bot_{\widetilde{\sigma}}})
=\sum_{i=1}^{k}\mathrm{dim}_{\mathbb{F}_{q}}(\mathcal{C}_{\varrho(i)}^{\bot_{\widetilde{\sigma}}}).\\
&\Longleftrightarrow\mathcal{C}_{i}\cap\mathcal{C}_{\varrho(i)}^{\bot_{\widetilde{\sigma}}}=\mathcal{C}_{\varrho(i)}^{\bot_{\widetilde{\sigma}}} \ \text{for $i=1,2,\ldots,k$}.\\
&\Longleftrightarrow\mathcal{C}_{\varrho(i)}^{\bot_{\widetilde{\sigma}}}\subseteq\mathcal{C}_{i} \ \text{for $i=1,2,\ldots,k$}.
\end{align*}
This completes the proof of statement (1). Similarly, we can prove statement (2). $\hfill\square$
\end{proof}

\vspace{6pt}

In the following remark, we elaborate on the advantages of Theorem \ref{theorem3.2} (1) compared to existing conclusions.

\begin{remark}
In Section 5 of the recent paper \cite{Cao2024Onthe}, the authors constructed several families of $\sigma$ dual-containing MP codes over some small finite fields. However, a systematic method for constructing $\sigma$ dual-containing MP codes was lacking in \cite[Section 5]{Cao2024Onthe}.
In contrast, Theorem \ref{theorem3.2} presents an explicit sufficient and necessary condition for MP codes to be $\sigma$ dual-containing;
in other words, it establishes a general method for constructing $\sigma$ dual-containing MP codes.
We point out that in \cite[Section 5]{Cao2024Onthe}, when constructing $\sigma$ dual-containing MP codes $\mathcal{C}(A)$, the defining matrices $A$ therein satisfy the condition that $\pi_{s}(A)M_{\widehat{\tau}}A^{\top}$ is an invertible diagonal matrix. In contrast, the defining matrices $A$ in Theorem \ref{theorem3.2} satisfy a more general condition, namely, $\pi_{s}(A)M_{\widehat{\tau}}A^{\top}$ is a $\varrho$-monomial matrix for any permutation
$\varrho\in S_{k}$. This leads to the following benefits:
\vspace{-4pt}
\begin{itemize}
\item [(1)] The variety of matrices $M_{\widehat{\tau}}$ and $M_{\tau}$ satisfying the conditions of Theorem \ref{theorem3.2} is far greater than those satisfying the conditions from \cite[Section 5]{Cao2024Onthe}. Therefore, the variety of $\sigma$ dual-containing MP codes constructed using the method described in Theorem \ref{theorem3.2} is much greater than those constructed using the method described in \cite[Section 5]{Cao2024Onthe}.

\vspace{-4pt}

\item [(2)] When constructing $\sigma$ dual-containing MP codes, \cite[Section 5]{Cao2024Onthe} requires each constituent code $\mathcal{C}_{i}$ to be dual-containing. However, Theorem \ref{theorem3.2} suggests that this condition may not be necessary for constructing $\sigma$ dual-containing MP codes.
    Specifically, the conditions regarding each constituent code $\mathcal{C}_{i}$ in Theorem \ref{theorem3.2} are highly flexible, satisfying $\mathcal{C}_{\rho(i)}^{\bot_{\widetilde{\sigma}}}\subseteq \mathcal{C}_{i}$ for an arbitrary permutation $\varrho \in S_{k}$. This implies that theoretically, the inclusion relationships between these constituent codes can reach up to $k!$ varieties. Therefore, the selection of constituent codes $\mathcal{C}_{i}$ is highly flexible when constructing $\sigma$ dual-containing MP codes using Theorem \ref{theorem3.2}.

\vspace{-4pt}

\item [(3)] It is not difficult to observe that the number of constituent codes $\mathcal{C}_{i}$ satisfying the relationship $\mathcal{C}_{\rho(i)}^{\bot_{\widetilde{\sigma}}}\subseteq \mathcal{C}_{i}$ is greater than the number of constituent codes $\mathcal{C}_{i}$ satisfying dual-containing. Therefore, using Theorem \ref{theorem3.2}, we can construct many new $\sigma$ dual-containing MP codes that cannot be obtained through the method in \cite{Cao2024Onthe}.
\end{itemize}
\end{remark}

\section{Constructing EAQECCs using $\sigma$ hulls of MP codes }

\begin{lemma}\label{lemma3.1}
Let $\mathcal{C}_{i}$ be an $[n,k_{i}]_{q}$ linear code with generator matrix $G_{i}$ and parity check matrix $H_{i}$ for $i=1,2$.
Set $\sigma=(\tau,\pi_{s})\in\mathrm{SLAut}(\mathbb{F}_{q}^{n})$, where $\tau$ corresponds to a monomial matrix $M_{\tau}=DP_{\tau}\in\mathbb{F}_{q}^{n\times n}$ and $1\leq s\leq e$. Then
\vspace{-5pt}
\begin{itemize}
\item [(1)] $\mathrm{dim}_{\mathbb{F}_{q}}(\mathcal{C}_{1}\cap\mathcal{C}_{2}^{\bot_{\sigma}})
=n-k_{2}-\mathrm{rank}(H_{1}M_{\tau}^{-1}\pi_{s}(H_{2})^{\top})
=k_{1}-\mathrm{rank}(\pi_{s}(G_{2})M_{\tau}G_{1}^{\top})$;

\vspace{-4pt}

\item [(2)] $\mathrm{dim}_{\mathbb{F}_{q}}((\mathcal{C}_{1}^{\bot_{\sigma}})^{\bot_{\sigma}}\cap\mathcal{C}_{2}^{\bot_{\sigma}})
=n-k_{2}-\mathrm{rank}(H_{2}M_{\tau}^{-1}\pi_{s}(H_{1})^{\top})
=k_{1}-\mathrm{rank}(\pi_{s}(G_{1})M_{\tau}G_{2}^{\top})$.
\end{itemize}
\end{lemma}

\begin{proof}
(1) By Lemmas \ref{lemma1} and \ref{lemma2} (1), we obtain
\vspace{-6pt}
\begin{align*}
\mathrm{dim}_{\mathbb{F}_{q}}(\mathcal{C}_{1}\cap\mathcal{C}_{2}^{\bot_{\sigma}})
=n-k_{2}-\mathrm{rank}(H_{1}(\pi_{s}(H_{2})(M_{\tau}^{-1})^{\top})^{\top})
=n-k_{2}-\mathrm{rank}(H_{1}M_{\tau}^{-1}\pi_{s}(H_{2})^{\top}).
\end{align*}
On the other hand, since $\mathcal{C}_{2}^{\bot_{\sigma}}=\sigma(\mathcal{C}_{2})^{\bot}$,
then $\mathcal{C}_{2}^{\bot_{\sigma}}$ has a parity check matrix $\sigma(G_{2})$, namely, $\pi_{s}(G_{2})M_{\tau}$.
Using Lemma \ref{lemma2} (2), we obtain
\vspace{-6pt}
\begin{align*}
\mathrm{dim}_{\mathbb{F}_{q}}(\mathcal{C}_{1}\cap\mathcal{C}_{2}^{\bot_{\sigma}})
=k_{1}-\mathrm{rank}(G_{1}(\pi_{s}(G_{2})M_{\tau})^{\top})
=k_{1}-\mathrm{rank}(\pi_{s}(G_{2})M_{\tau}G_{1}^{\top}).
\end{align*}

(2) Since $(\mathcal{C}_{1}^{\bot_{\sigma}})^{\bot_{\sigma}}=\sigma(\mathcal{C}_{1}^{\bot_{\sigma}})^{\bot}$, we deduce from Lemma \ref{lemma1} that
$(\mathcal{C}_{1}^{\bot_{\sigma}})^{\bot_{\sigma}}$ has a parity check matrix $\sigma(\pi_{s}(H_{1})(M_{\tau}^{-1})^{\top})$.
Using Lemma \ref{lemma2} (1), we obtain
\vspace{-6pt}
\begin{align*}
\mathrm{dim}_{\mathbb{F}_{q}}((\mathcal{C}_{1}^{\bot_{\sigma}})^{\bot_{\sigma}}\cap\mathcal{C}_{2}^{\bot_{\sigma}})
&=n-k_{2}-\mathrm{rank}(\sigma(\pi_{s}(H_{1})(M_{\tau}^{-1})^{\top})(\pi_{s}(H_{2})(M_{\tau}^{-1})^{\top})^{\top})\\
&=n-k_{2}-\mathrm{rank}(\pi_{s}(\pi_{s}(H_{1})(M_{\tau}^{-1})^{\top})M_{\tau}(\pi_{s}(H_{2})(M_{\tau}^{-1})^{\top})^{\top})\\
&=n-k_{2}-\mathrm{rank}(\pi_{s}(\pi_{s}(H_{1})(M_{\tau}^{-1})^{\top})\pi_{s}(H_{2})^{\top})\\
&=n-k_{2}-\mathrm{rank}(\pi_{s}(H_{1})(M_{\tau}^{-1})^{\top}H_{2}^{\top})\\
&=n-k_{2}-\mathrm{rank}(H_{2}M_{\tau}^{-1}\pi_{s}(H_{1})^{\top}).
\end{align*}

On the other hand, since $\mathcal{C}_{1}^{\bot_{\sigma}}$ has a parity check matrix $\sigma(G_{1})$,
we deduce from Lemma \ref{lemma1} that $(\mathcal{C}_{1}^{\bot_{\sigma}})^{\bot_{\sigma}}$ has a generator matrix
$\pi_{s}(\sigma(G_{1}))(M_{\tau}^{-1})^{\top}$. Using Lemma \ref{lemma2} (2), we obtain
\vspace{-6pt}
\begin{align*}
\mathrm{dim}_{\mathbb{F}_{q}}((\mathcal{C}_{1}^{\bot_{\sigma}})^{\bot_{\sigma}}\cap\mathcal{C}_{2}^{\bot_{\sigma}})
&=k_{1}-\mathrm{rank}(\pi_{s}(\sigma(G_{1}))(M_{\tau}^{-1})^{\top}\sigma(G_{2})^{\top})\\
&=k_{1}-\mathrm{rank}(\pi_{s}(\pi_{s}(G_{1})M_{\tau})(M_{\tau}^{-1})^{\top}(\pi_{s}(G_{2})M_{\tau})^{\top})\\
&=k_{1}-\mathrm{rank}(\pi_{s}(\pi_{s}(G_{1})M_{\tau})\pi_{s}(G_{2})^{\top})\\
&=k_{1}-\mathrm{rank}(\pi_{s}(G_{1})M_{\tau}G_{2}^{\top}). \tag*{\ensuremath{\square}}
\end{align*}
\end{proof}

\begin{corollary}\label{corollary3.2}
Let $\mathcal{C}$ be an $[n,k]_{q}$ linear code with generator matrix $G$ and parity check matrix $H$. Set $\sigma=(\tau,\pi_{s})\in\mathrm{SLAut}(\mathbb{F}_{q}^{n})$,
where $\tau$ corresponds to a monomial matrix $M_{\tau}=DP_{\tau}\in\mathbb{F}_{q}^{n\times n}$ and $1\leq s\leq e$. Then
\begin{align*}
\mathrm{dim}_{\mathbb{F}_{q}}(\mathrm{Hull}_{\sigma}(\mathcal{C}))
=\mathrm{dim}_{\mathbb{F}_{q}}(\mathrm{Hull}_{\sigma}(\mathcal{C}^{\bot_{\sigma}}))=n-k-\mathrm{rank}(HM_{\tau}^{-1}\pi_{s}(H)^{\top})
=k-\mathrm{rank}(\pi_{s}(G)M_{\tau}G^{\top}).
\end{align*}
\end{corollary}

\begin{proof}
In Lemma \ref{lemma3.1}, by setting $\mathcal{C}_{1}=\mathcal{C}_{2}=\mathcal{C}=[n,k]_{q}$ with generator matrix $G$ and parity check matrix $H$,
the result then follows. $\hfill\square$
\end{proof}

\vspace{6pt}

The following result is a corollary of \cite[Theorem 4]{Galindo2019Entanglement-assisted}, demonstrating the method for generating $q$-ary EAQECCs from $q$-ary classical linear codes.

\begin{lemma}\label{lemma3.3}
(\!\! \cite{Galindo2019Entanglement-assisted})
Let $\mathcal{C}_{i}$ be an $[n,k_{i},d_{i}]_{q}$ linear code with parity check matrix $H_{i}$ for $i=1,2$. Then there exists an
$[[n,k_{1}+k_{2}-n+c,\min\{d_{1},d_{2}\};c]]_{q}$ EAQECC, where $c=\mathrm{rank}(H_{1}H_{2}^{\top})$ is the required number of maximally entangled states.
\end{lemma}

Let us now present the general method for constructing two types of EAQECCs by using the $\sigma$ hulls of linear codes in the following theorem.

\begin{theorem}\label{theorem4.4}
Let $\mathcal{C}$ be an $[n,k,d]_{q}$ linear code whose $\sigma$ dual code has minimum distance $d'$. Then the following statements hold.
\vspace{-5pt}
\begin{itemize}
\item [(1)] There exists an
$[[n,k-\mathrm{dim}_{\mathbb{F}_{q}}(\mathrm{Hull}_{\sigma}(\mathcal{C})),d;n-k-\mathrm{dim}_{\mathbb{F}_{q}}(\mathrm{Hull}_{\sigma}(\mathcal{C}))]]_{q}$ EAQECC;

\vspace{-4pt}

\item [(2)] There exists an
$[[n,n-k-\mathrm{dim}_{\mathbb{F}_{q}}(\mathrm{Hull}_{\sigma}(\mathcal{C})),d';k-\mathrm{dim}_{\mathbb{F}_{q}}(\mathrm{Hull}_{\sigma}(\mathcal{C}))]]_{q}$ EAQECC.
\end{itemize}
\end{theorem}

\begin{proof}
(1) In Lemma \ref{lemma3.3}, let $\mathcal{C}_{1}=\mathcal{C}$ be an $[n,k,d]_{q}$ with parity check matrix $H$, and let $\mathcal{C}_{2}=\sigma(\mathcal{C})$.
Then, it follows from $\mathcal{C}^{\bot_{\sigma}}=(\sigma(\mathcal{C}))^{\bot}$ and Lemma \ref{lemma1} that $\pi_{s}(H)(M_{\tau}^{-1})^{\top}$ is a parity check matrix of $\mathcal{C}_{2}$.
By Lemma \ref{lemma3.3}, we can obtain an $[[n,2k-n+c,d;c]]_{q}$ EAQECC, with the required number of maximally entangled states determined as follows:
\vspace{-4pt}
\begin{align*}
c&=\mathrm{rank}(H(\pi_{s}(H)(M_{\tau}^{-1})^{\top})^{\top})\\
&=\mathrm{rank}(HM_{\tau}^{-1}\pi_{s}(H)^{\top})\\
&=n-k-\mathrm{dim}_{\mathbb{F}_{q}}(\mathrm{Hull}_{\sigma}(\mathcal{C})),
\end{align*}
where the third equation follows from Corollary \ref{corollary3.2}.
Substituting it into the parameters of the $[[n,2k-n+c,d;c]]_{q}$ EAQECC, we can obtain the EAQECC described in statement (1).

(2) Statement (1) reveals that we can derive an
$[[n,k-\mathrm{dim}_{\mathbb{F}_{q}}(\mathrm{Hull}_{\sigma}(\mathcal{C})),d;n-k-\mathrm{dim}_{\mathbb{F}_{q}}(\mathrm{Hull}_{\sigma}(\mathcal{C}))]]_{q}$ EAQECC
from an $[n,k,d]_{q}$ linear code $\mathcal{C}$.
By replacing $\mathcal{C}$ with $\mathcal{C}^{\bot_{\sigma}}$, which has a minimum distance of $d'$, we can then obtain an
$[[n,n-k-\mathrm{dim}_{\mathbb{F}_{q}}(\mathrm{Hull}_{\sigma}(\mathcal{C}^{\bot_{\sigma}})),d';
k-\mathrm{dim}_{\mathbb{F}_{q}}(\mathrm{Hull}_{\sigma}(\mathcal{C}^{\bot_{\sigma}}))]]_{q}$ EAQECC.
Considering $\mathrm{dim}_{\mathbb{F}_{q}}(\mathrm{Hull}_{\sigma}(\mathcal{C}))=\mathrm{dim}_{\mathbb{F}_{q}}(\mathrm{Hull}_{\sigma}(\mathcal{C}^{\bot_{\sigma}}))$ from Corollary \ref{corollary3.2}, we therefore obtain the EAQECC described in statement (2). $\hfill\square$
\end{proof}

\vspace{6pt}

Next, we give a useful result, which will be used for constructing EAQECCs with flexible parameters.

\begin{theorem}\label{theorem4.5}
Let $q>2$ and let $\mathcal{C}_{i}$ be an $[n,k_{i}]_{q}$ linear code for $i=1,2$. Set $\sigma=(\tau,\pi_{s})\in\mathrm{SLAut}(\mathbb{F}_{q}^{n})$, where $\tau$ corresponds to a monomial matrix $M_{\tau}=DP_{\tau}\in\mathbb{F}_{q}^{n\times n}$ and $1\leq s\leq e$. Then for any integer $h$ with
$\mathrm{max}\{0,k_{1}-k_{2}\}\leq h\leq \mathrm{dim}_{\mathbb{F}_{q}}(\mathcal{C}_{1}\cap\mathcal{C}_{2}^{\bot_{\sigma}})$, there exists a linear code
$\mathcal{C}_{2,h}$ monomially equivalent to $\mathcal{C}_{2}$ such that
\begin{align*}
\mathrm{dim}_{\mathbb{F}_{q}}(\mathcal{C}_{1}\cap\mathcal{C}_{2,h}^{\bot_{\sigma}})=h.
\end{align*}
\end{theorem}

\begin{proof}
Since
\begin{align*}
\mathrm{dim}_{\mathbb{F}_{q}}(\mathcal{C}_{1}\cap\mathcal{C}_{2}^{\bot_{\sigma}})
=\mathrm{dim}_{\mathbb{F}_{q}}(\mathcal{C}_{1}\cap\sigma(\mathcal{C}_{2})^{\bot})
=\mathrm{dim}_{\mathbb{F}_{q}}(\mathcal{C}_{1}\cap(\pi_{s}(\mathcal{C}_{2})M_{\tau})^{\bot}),
\end{align*}
then by using Lemma \ref{lemma2.6}, there exists a monomial matrix $M_{\tau'}$ such that
\begin{align*}
\mathrm{dim}_{\mathbb{F}_{q}}(\mathcal{C}_{1}\cap(\pi_{s}(\mathcal{C}_{2})M_{\tau}M_{\tau'})^{\bot})
=h.
\end{align*}
Putting $M_{\tau''}=\pi_{e-s}(M_{\tau}M_{\tau'}M_{\tau}^{-1})$, then $M_{\tau''}$ is monomial, and $\pi_{s}(M_{\tau''})M_{\tau}=M_{\tau}M_{\tau'}$.
Taking $\mathcal{C}_{2,h}=\mathcal{C}_{2}M_{\tau''}$, then $\mathcal{C}_{2,h}$ is monomially equivalent to $\mathcal{C}_{2}$, and we have
\vspace{-6pt}
\begin{align*}
\mathrm{dim}_{\mathbb{F}_{q}}(\mathcal{C}_{1}\cap\mathcal{C}_{2,h}^{\bot_{\sigma}})
&=\mathrm{dim}_{\mathbb{F}_{q}}(\mathcal{C}_{1}\cap(\pi_{s}(\mathcal{C}_{2,h})M_{\tau})^{\bot})\\
&=\mathrm{dim}_{\mathbb{F}_{q}}(\mathcal{C}_{1}\cap(\pi_{s}(\mathcal{C}_{2}M_{\tau''})M_{\tau})^{\bot})\\
&=\mathrm{dim}_{\mathbb{F}_{q}}(\mathcal{C}_{1}\cap(\pi_{s}(\mathcal{C}_{2})\pi_{s}(M_{\tau''})M_{\tau})^{\bot})\\
&=\mathrm{dim}_{\mathbb{F}_{q}}(\mathcal{C}_{1}\cap(\pi_{s}(\mathcal{C}_{2})M_{\tau}M_{\tau'})^{\bot})\\
&=h,
\end{align*}
which completes the proof. $\hfill\square$
\end{proof}

\begin{corollary}\label{corollary4.6}
Let $q>2$ and let $\mathcal{C}$ be an $[n,k,d]_{q}$ linear code. Set $\sigma=(\tau,\pi_{s})\in\mathrm{SLAut}(\mathbb{F}_{q}^{n})$, where $\tau$ corresponds to a monomial matrix $M_{\tau}=DP_{\tau}\in\mathbb{F}_{q}^{n\times n}$ and $1\leq s\leq e$. Then
there exists a monomially equivalent $[n,k,d]_{q}$ linear code $\mathcal{C}_{h}$ with
$\mathrm{dim}_{\mathbb{F}_{q}}(\mathrm{Hull}_{\sigma}(\mathcal{C}_{h}))=h$ for every integer $h$ with
$0\leq h\leq \mathrm{dim}_{\mathbb{F}_{q}}(\mathrm{Hull}_{\sigma}(\mathcal{C}))$.
\end{corollary}

Combining Theorem \ref{theorem4.4} and Corollary \ref{corollary4.6}, we directly obtain the following theorem.

\begin{theorem}\label{theorem4.7}
Let $q>2$ and let $\mathcal{C}$ be an $[n,k,d]_{q}$ linear code. Set $\sigma=(\tau,\pi_{s})\in\mathrm{SLAut}(\mathbb{F}_{q}^{n})$, where $\tau$ corresponds to a monomial matrix $M_{\tau}=DP_{\tau}\in\mathbb{F}_{q}^{n\times n}$ and $1\leq s\leq e$. Then there exists a monomially equivalent $[n,k,d]_{q}$ linear code $\mathcal{C}_{h}$, whose $\sigma$ dual code has minimum distance $d'$, such that the following statements hold.
\vspace{-5pt}
\begin{itemize}
\item [(1)] There exists an
$[[n,k-h,d;n-k-h]]_{q}$ EAQECC for every integer $h$ with $0\leq h\leq \mathrm{dim}_{\mathbb{F}_{q}}(\mathrm{Hull}_{\sigma}(\mathcal{C}))$;

\vspace{-4pt}

\item [(2)] There exists an
$[[n,n-k-h,d';k-h]]_{q}$ EAQECC for every integer $h$ with $0\leq h\leq \mathrm{dim}_{\mathbb{F}_{q}}(\mathrm{Hull}_{\sigma}(\mathcal{C}))$.
\end{itemize}
\end{theorem}

\begin{corollary}
Let $q>2$ and let $\mathcal{C}$ be an $[n,k,n-k+1]_{q}$ MDS linear code. Set $\sigma=(\tau,\pi_{s})\in\mathrm{SLAut}(\mathbb{F}_{q}^{n})$, where $\tau$ corresponds to a monomial matrix $M_{\tau}=DP_{\tau}\in\mathbb{F}_{q}^{n\times n}$ and $1\leq s\leq e$. Then the following statements hold.
\vspace{-5pt}
\begin{itemize}
\item [(1)] There exists an
$[[n,k-h,n-k+1;n-k-h]]_{q}$ MDS EAQECC for every integer $h$ with $0\leq h\leq \mathrm{dim}_{\mathbb{F}_{q}}(\mathrm{Hull}_{\sigma}(\mathcal{C}))$;

\vspace{-4pt}

\item [(2)] There exists an
$[[n,n-k-h,k+1;k-h]]_{q}$ MDS EAQECC for every integer $h$ with $0\leq h\leq \mathrm{dim}_{\mathbb{F}_{q}}(\mathrm{Hull}_{\sigma}(\mathcal{C}))$.
\end{itemize}
\end{corollary}

Next, let us present a general construction method for deriving EAQECCs with flexible parameters from MP codes.

\begin{theorem}
Set $\sigma=(\tau,\pi_{s})\in\mathrm{SLAut}(\mathbb{F}_{q}^{kn})$ and $\widetilde{\sigma}=(\widetilde{\tau},\pi_{s})\in\mathrm{SLAut}(\mathbb{F}_{q}^{n})$,
where $1\leq s\leq e$, $\widetilde{\tau}$ corresponds to a $\widetilde{\tau}$-monomial matrix $M_{\widetilde{\tau}}\in\mathbb{F}_{q}^{n}$,
and $\tau$ corresponds to a $\tau$-monomial matrix $M_{\tau}=M_{\widehat{\tau}}\otimes M_{\widetilde{\tau}}\in\mathbb{F}_{q}^{kn}$ for $\widehat{\tau}$-monomial matrix $M_{\widehat{\tau}}\in\mathbb{F}_{q}^{k}$.
Let $\mathcal{C}(A)=[\mathcal{C}_{1},\mathcal{C}_{2},\ldots,\mathcal{C}_{k}]\cdot A$, where $A\in\mathbb{F}_{q}^{k\times k}$ and $\mathcal{C}_{i}$ is an $[n,t_{i},d_{i}]_{q}$ linear code whose $\widetilde{\sigma}$ dual code has minimum distance $d_{i}'$ for $i=1,2,\ldots,k$.
If $\pi_{s}(A)M_{\widehat{\tau}}A^{\top}$ is $\varrho$-monomial for some permutation $\varrho\in S_{k}$, then the following two statements hold.
\vspace{-5pt}
\begin{itemize}
\item [(1)] There exists a $q$-ary EAQECC $\mathcal{Q}_{1}: \big[\mspace{-3mu}\big[kn,\sum_{i=1}^{k}t_{i}-h,\geq\min\limits_{\scriptscriptstyle 1\leq i\leq k}\{D_{i}(A)d_{i}\};kn-\sum_{i=1}^{k}t_{i}-h\big]\mspace{-3mu}\big]_{q}$ for every integer $h$ with
$0\leq h\leq \sum_{i=1}^{k}\mathrm{dim}_{\mathbb{F}_{q}}(\mathcal{C}_{i}\cap\mathcal{C}_{\rho(i)}^{\bot_{\widetilde{\sigma}}})$.

\vspace{-4pt}

\item [(2)] There exists a $q$-ary EAQECC $\mathcal{Q}_{2}: \big[\mspace{-3mu}\big[kn,kn-\sum_{i=1}^{k}t_{i}-h,\geq\min\limits_{\scriptscriptstyle 1\leq i\leq k}\{D_{i}(A)d_{\varrho(i)}'\};\sum_{i=1}^{k}t_{i}-h\big]\mspace{-3mu}\big]_{q}$ for every integer $h$ with
$0\leq h\leq \sum_{i=1}^{k}\mathrm{dim}_{\mathbb{F}_{q}}(\mathcal{C}_{i}\cap\mathcal{C}_{\rho(i)}^{\bot_{\widetilde{\sigma}}})$.
\end{itemize}
\end{theorem}

\begin{proof}
(1) By Theorem \ref{theorem3.1}, we know that $\mathcal{C}(A)$ is a
$\big[kn,\sum_{i=1}^{k}t_{i},\geq\min\limits_{\scriptscriptstyle 1\leq i\leq k}\{D_{i}(A)d_{i}\}\big]_{q}$ linear code with
$\mathrm{dim}_{\mathbb{F}_{q}}(\mathrm{Hull}_{\sigma}(\mathcal{C}(A)))
=\sum_{i=1}^{k}\mathrm{dim}_{\mathbb{F}_{q}}(\mathcal{C}_{i}\cap\mathcal{C}_{\rho(i)}^{\bot_{\widetilde{\sigma}}})$.
Then by Corollary \ref{corollary4.6}, there exists a monomially equivalent
$\big[kn,\sum_{i=1}^{k}t_{i},\geq\min\limits_{\scriptscriptstyle 1\leq i\leq k}\{D_{i}(A)d_{i}\}\big]_{q}$ linear code $\mathcal{C}_{h}$
with $\mathrm{dim}_{\mathbb{F}_{q}}(\mathrm{Hull}_{\sigma}(\mathcal{C}_{h}))=h$ for every integer $h$ with
$0\leq h\leq \mathrm{dim}_{\mathbb{F}_{q}}(\mathrm{Hull}_{\sigma}(\mathcal{C}(A)))
=\sum_{i=1}^{k}\mathrm{dim}_{\mathbb{F}_{q}}(\mathcal{C}_{i}\cap\mathcal{C}_{\rho(i)}^{\bot_{\widetilde{\sigma}}})$.
By Theorem \ref{theorem4.7}, we obtain a $q$-ary EAQECC $\mathcal{Q}_{1}: \big[\mspace{-3mu}\big[kn,\sum_{i=1}^{k}t_{i}-h,\geq\min\limits_{\scriptscriptstyle 1\leq i\leq k}\{D_{i}(A)d_{i}\};kn-\sum_{i=1}^{k}t_{i}-h\big]\mspace{-3mu}\big]_{q}$ for every integer $h$ with
$0\leq h\leq\sum_{i=1}^{k}\mathrm{dim}_{\mathbb{F}_{q}}(\mathcal{C}_{i}\cap\mathcal{C}_{\rho(i)}^{\bot_{\widetilde{\sigma}}})$.
This completes the proof of statement (1).

\vspace{6pt}

(2) As $\pi_{s}(A)M_{\widehat{\tau}}A^{\top}$ is $\varrho$-monomial for some permutation $\varrho\in S_{k}$, it can be written as
$\pi_{s}(A)M_{\widehat{\tau}}A^{\top}=D_{\varrho}P_{\varrho}$, where $D_{\varrho}$ is invertible diagonal and $P_{\varrho}$ is a permutation matrix under $\varrho$. Then $(M_{\widehat{\tau}}^{\top}\pi_{s}(A)^{\top})^{-1}=D_{\varrho}^{-1}P_{\varrho}A$.
Combining this with \cite[Remark 4.5 2)]{Cao2024Onthe}, we obtain
\begin{align*}
\mathcal{C}(A)^{\bot_{\sigma}}
&=[\mathcal{C}_{1}^{\bot_{\widetilde{\sigma}}},\mathcal{C}_{2}^{\bot_{\widetilde{\sigma}}},\ldots,\mathcal{C}_{k}^{\bot_{\widetilde{\sigma}}}]\cdot (M_{\widehat{\tau}}^{\top}\pi_{s}(A)^{\top})^{-1}\\
&=[\mathcal{C}_{1}^{\bot_{\widetilde{\sigma}}},\mathcal{C}_{2}^{\bot_{\widetilde{\sigma}}},\ldots,\mathcal{C}_{k}^{\bot_{\widetilde{\sigma}}}]\cdot D_{\varrho}^{-1}P_{\varrho}A\\
&=[\mathcal{C}_{1}^{\bot_{\widetilde{\sigma}}},\mathcal{C}_{2}^{\bot_{\widetilde{\sigma}}},\ldots,\mathcal{C}_{k}^{\bot_{\widetilde{\sigma}}}]\cdot P_{\varrho}A\\
&=[\mathcal{C}_{\varrho(1)}^{\bot_{\widetilde{\sigma}}},\mathcal{C}_{\varrho(2)}^{\bot_{\widetilde{\sigma}}},\ldots,
\mathcal{C}_{\varrho(k)}^{\bot_{\widetilde{\sigma}}}]\cdot A.
\end{align*}
By Corollary \ref{corollary3.2}, we obtain that $\mathcal{C}(A)^{\bot_{\sigma}}$ is a
$\big[kn,kn-\sum_{i=1}^{k}t_{i},\geq\min\limits_{\scriptscriptstyle 1\leq i\leq k}\{D_{i}(A)d_{\varrho(i)}'\}\big]_{q}$ linear code with
\vspace{-4pt}
\begin{align*}
\mathrm{dim}_{\mathbb{F}_{q}}(\mathrm{Hull}_{\sigma}(\mathcal{C}(A)^{\bot_{\sigma}}))
=\mathrm{dim}_{\mathbb{F}_{q}}(\mathrm{Hull}_{\sigma}(\mathcal{C}(A)))
=\sum_{i=1}^{k}\mathrm{dim}_{\mathbb{F}_{q}}(\mathcal{C}_{i}\cap\mathcal{C}_{\rho(i)}^{\bot_{\widetilde{\sigma}}}).
\end{align*}
By Corollary \ref{corollary4.6}, there exists a monomially equivalent
$\big[kn,kn-\sum_{i=1}^{k}t_{i},\geq\min\limits_{\scriptscriptstyle 1\leq i\leq k}\{D_{i}(A)d_{\varrho(i)}'\}\big]_{q}$ linear code $\mathcal{C}_{h}$ with $\mathrm{dim}_{\mathbb{F}_{q}}(\mathrm{Hull}_{\sigma}(\mathcal{C}_{h}))=h$ for every integer $h$ with
$0\leq h\leq \mathrm{dim}_{\mathbb{F}_{q}}(\mathrm{Hull}_{\sigma}(\mathcal{C}(A)^{\bot_{\sigma}}))
=\sum_{i=1}^{k}\mathrm{dim}_{\mathbb{F}_{q}}(\mathcal{C}_{i}\cap\mathcal{C}_{\rho(i)}^{\bot_{\widetilde{\sigma}}})$.
Using Theorem \ref{theorem4.7}, we obtain a $q$-ary EAQECC $\mathcal{Q}_{2}: \big[\mspace{-3mu}\big[kn,kn-\sum_{i=1}^{k}t_{i}-h,\geq\min\limits_{\scriptscriptstyle 1\leq i\leq k}\{D_{i}(A)d_{\varrho(i)}'\};\sum_{i=1}^{k}t_{i}-h\big]\mspace{-3mu}\big]_{q}$ for every integer $h$ with
$0\leq h\leq \sum_{i=1}^{k}\mathrm{dim}_{\mathbb{F}_{q}}(\mathcal{C}_{i}\cap\mathcal{C}_{\rho(i)}^{\bot_{\widetilde{\sigma}}})$.
This completes the proof of statement (2). $\hfill\square$
\end{proof}

\section*{Data availability statement}

No data was used for the research described in this paper.

\footnotesize{
\bibliographystyle{plain}
\phantomsection
\addcontentsline{toc}{section}{References}
\bibliography{sigma}
}

\end{document}